\documentclass[10pt,twocolumn,twoside]{IEEEtran}
\IEEEoverridecommandlockouts
\usepackage{cite}
\usepackage{amsmath,amssymb,amsfonts}
\usepackage{graphicx}
\usepackage{textcomp}
\usepackage{amsthm}
\usepackage{changepage}
\usepackage{array}
\usepackage{hyperref}
\usepackage{dsfont}
\usepackage{bm}
\usepackage{caption}

\usepackage{soul}
\usepackage{color}

\usepackage{subcaption}
\usepackage{enumitem}
\usepackage{booktabs, tabularx}
\usepackage{multirow}
\usepackage{algorithm}
\usepackage{algpseudocode}
\algrenewcommand\algorithmicrequire{\textbf{Input:}}
\algrenewcommand\algorithmicensure{\textbf{Output:}}
\usepackage{float}
\usepackage{subfiles}
\usepackage{verbatim}
\usepackage{bbm}

\newtheorem{theorem}{Theorem}
\newtheorem{definition}{Definition}
\newtheorem{corollary}{Corollary}
\usepackage{cuted}
\usepackage{stfloats}
\newtheorem{proposition}{Proposition}
\newtheorem{lemma}{Lemma}
\usepackage[font = small, justification = justified,format = plain]{caption}
\DeclareMathOperator*{\argmax}{arg\,max}

\newcommand{\pip}{\pi^+}
\newcommand{\pim}{\pi^-}

\newcommand{\bp}{\pmb p}
\newcommand{\bc}{\pmb c}
\newcommand{\bd}{\pmb d}
\newcommand{\bx}{\pmb x}
\newcommand{\btau}{\pmb \tau}
\newcommand{\br}{\pmb r}

\newcommand{\R}{\mathbb{R}}
\newcommand{\clip}{\mathrm{clip}}

\theoremstyle{plain}
\makeatletter
\thm@space@setup{\thm@preskip=3pt \thm@postskip=3pt}
\makeatother

\algrenewcommand\algorithmicindent{0.5em}
\usepackage{xcolor}
\def\BibTeX{{\rm B\kern-.05em{\sc i\kern-.025em b}\kern-.08em
    T\kern-.1667em\lower.7ex\hbox{E}\kern-.125emX}}
\begin{document}

\title{Threshold Pricing for Distributed Scheduling of Flexible Demands in Energy Communities}
\author{Minjae~Jeon,~\IEEEmembership{Student Member,~IEEE,}
        Lang~Tong,~\IEEEmembership{Fellow,~IEEE,}
        Qing~Zhao,~\IEEEmembership{Fellow,~IEEE}
    \thanks{Minjae Jeon, Lang Tong, and Qing Zhao are with the School of Electrical and Computer Engineering, Cornell University, USA. This work was supported in part by the National Science Foundation under Grant 2412776, 2419622, and 2603293.}
    }

\maketitle
\begin{abstract}
This paper develops a price-based distributed scheduling in an energy community whose members own behind-the-meter renewable generation with deferrable EV charging and price-elastic thermostatic loads. A coordinator transacts with the distribution utility under a Net Energy Metering tariff and broadcasts a community price to which each household responds in its own interest, giving a bilevel stochastic dynamic program that is intractable in general. Our main result characterizes that the joint optimal centralized policy is a two-threshold policy on the community's aggregate renewable generation. Building on this structure, we adopt the Threshold Pricing Rule, which is uniform, individually rational, revenue adequate, and asymptotically optimal in terms of community welfare under a light-traffic condition. Simulations using synthetic and real world data confirm asymptotic optimality and individual surplus gains.

\end{abstract}

\section{Introduction}
We consider price-based distributed scheduling problem in an energy community that pools households with behind-the-meter resources and transacts with the distribution utility as a single customer. The energy community is under Net Energy Metering (NEM) tariff that charges the retail rate $\pip$ on its net-import and credits the compensation rate $\pim$ on its net-export, with $\pip > \pim$. A coordinator transacts with the utility on the households' behalf and sets the prices that households face internally. The coordinator's problem is to choose those prices so that households, each acting in its own interest, collectively schedule to maximize community's social welfare.

The households we consider carry local renewable generation and two kinds of flexible demand that need to be served. EV charging is deferrable: it can be shifted across intervals but must be completed by a deadline, which couples decisions over time. A thermostatically controlled load (TCL) is non-deferrable but price-elastic: it must be served within the interval, yet its consumption level responds to price through a comfort utility that is private to the household. 

Scheduling the two jointly with local renewables is what makes the pricing problem difficult. Each household solves a stochastic dynamic program for EV charging, while choosing a comfort level in each interval; the coordinator's pricing problem sits above all $N$ of these, giving a bilevel stochastic dynamic program that is intractable in general. The coordinator's pricing must further satisfy \textit{individual rationality}---the community price should be at least as favorable as NEM, or members will abandon the community---and \textit{revenue adequacy}---the price must allow the coordinator to cover its costs. 

\subsection{Related works}
The problem of meeting demand within an energy community has been studied extensively; see \cite{Barabino23Review} for a recent survey of community models and design objectives. We consider a setting with prosumers served by a centralized local market, which appears in comparable form in \cite{Han&Morstyn&McCulloch:19TPS,	chakraborty2018analysis, alahmed2024dynamic, li2024decentralized, mignoni2023distributed}. Yet none jointly addressed scheduling of deferrable and price-elastic non-deferrable demands with stochastic renewables.

Where prices coordinate these communities, two approaches have emerged, distinguished by how the coordinator derives the price. The first is iterative: the coordinator posts tentative prices, members report their tentative responses, and the exchange repeats until the prices converge \cite{gan2013optimal,gan2013real,ma2016efficient,Ma&Gupta&Topcu:17TAC,mignoni2023distributed,davoudi2025non}. These schemes solve static optimizations over forecast horizons and require repeated communication in every scheduling interval. The second is non-iterative: the coordinator sets the price directly from the structure of the retail tariff or a centralized optimization, and members respond once \cite{chakraborty2018analysis,alahmed2024dynamic,li2024decentralized}; \cite{alahmed2024dynamic} first introduced a three-zone pricing structure, the same structure our pricing policy adopts. The present work follows this second line. Both lines lack a price that coordinates members' dynamic decision problems. Existing rules either treat non-deferrable demand as inelastic or omit deferrable demand.

The closest work to consider a price that coordinates distributed dynamic programs is Hawkins's Lagrangian decomposition of weakly coupled dynamic programs \cite{Hawkins:03Thesis}, where dual variables serve as the pricing signal, but the approach does not extend readily to the uncountable state space considered here.

\subsection{Summary of results and contributions}
This work has three main contributions. First, we characterize the optimal centralized schedule when price-elastic TCL demand and deadline-constrained EV charging are scheduled jointly. The optimal policy is a two-threshold policy on the community's aggregate renewable generation, partitioning into \textit{net-consuming, net-zero, and net-producing} zones. The thresholds and the scheduling decisions in the net-consuming and net-producing zones are given in closed form, while the net-zero zone decision is characterized by the Bellman equation. 

Second, building on this structure, we adopt the Threshold Pricing Rule (TPR), which uses the same thresholds as the centralized optimal policy and broadcasts the NEM retail rate $\pip$ in the net-consuming zone, the compensation rate $\pim$ in the net-producing zone, and the NEM tariff in the net-zero zone. Given the broadcast price, we also obtain each household's TCL and EV decisions in closed form. We show that the TPR is individually rational, revenue adequate, and asymptotically optimal as community size grows under a light traffic condition in terms of community surplus. Lastly, numerical results using synthetic and real world EV charging and residential solar data confirm the asymptotic behavior and quantify the surplus each member gains within the community.

Preliminary results on price-based distributed scheduling in energy communities with the price-inelastic non-deferrable demands appear in \cite{jeon2026optimal}. We extend the pricing rule to joint scheduling with price-elastic non-deferrable demand, derived from the optimal centralized policy of the joint scheduling problem, and present the related numerical study. Proofs of theoretical results are in the appendix of \cite{JeonTongZhao-arxiv}

\section{Problem Formulation} \label{sec:cscheduling}
We consider a finite scheduling horizon $\mathcal{T}=\{1,\dots,T\}$ and a community of $N$ households indexed by $i \in [N]$. All variables are in kilowatt-hours (kWh) \footnote{Vectors are bold-faced letters, and $\mathbbm{1}(\cdot)$ represents the indicator function.}. The models below follow the formulation in preliminary results \cite{jeon2026optimal}.


\subsection{Household demand and energy resources}

\subsubsection{Deferrable EV demand with deadline}
Each household has a single EV charger serving one EV at a time. At the beginning of interval $t$, if the charger is unoccupied, a new EV arrives with probability $\alpha_{i,t}$, independently across households. Upon arrival, the vehicle's total energy demand $D_i$ and remaining time to deadline $T_i \le T-t+1$ are drawn from the joint PMF $\mathbb P_{i,t}$, with $\mathbb P_{i,t}(0, 0) = 1 - \alpha_{i,t}$\footnote{$D_i \le T_i \bar c$ without loss of generality.}.

The EV state is $(d_{i,t}, \tau_{i,t})$, the remaining demand and intervals to deadline, with $(0,0)$ denoting an idle charger:
\begin{equation}
	(d_{i,t+1},\tau_{i,t+1}) =
	\begin{cases}
		(D_i,T_i) \sim \mathbb{P}_{i,t}(\cdot), & (d_{i,t},\tau_{i,t})=(0,0),\\
		(d_{i,t},\tau_{i,t})-(c_{i,t},1), & \text{otherwise},
	\end{cases}
	\label{eq:ev-dynamics}
\end{equation}
where $c_{i,t}\in[0,\min\{d_{i,t},\bar{c}\}]$ and $\bar{c}$ is the per-interval charging limit. Unserved demand $d$ at the deadline incurs a penalty $q(d)$, where $q$ is strictly increasing and convex with $q'(0) > \pip$. The marginal penalty condition ensures that completing charging demand is always preferred to paying the penalty.  

\subsubsection{Price-responsive TCL demand}
The non-deferrable demand of household $i$ is a TCL whose consumption $p_{i,t}\in\mathbb{R}_+$ must be served within the interval but whose level adjusts to the price. The household's preference is captured by a private utility function $U_{i,t}(p_{i,t})$, assumed differentiable, concave, and time-varying, known only to household $i$. 
\subsubsection{Renewable generation}
Household $i$'s renewable generation $r_{i,t} \ge 0$ is revealed at the beginning of interval $t$ and evolves as a discrete-time Markov process on an uncountable state space with transition kernel $f_{i,t}$:
\begin{equation}
	r_{i,t+1}\sim f_{i,t}(\cdot\mid r_{i,t}),
	\qquad
	r_t := \textstyle\sum_{i\in[N]} r_{i,t},
	\label{eq:renewable}
\end{equation}
where generation is independent across households. 

The state of household $i$ is $\pmb x_{i,t}:=(d_{i,t},\tau_{i,t},r_{i,t})$, and the community state is $\pmb x_t:=(\pmb x_{1,t},\ldots,\pmb x_{N,t})$. 

\subsection{Net Consumption and Payments}
\subsubsection{Net consumption}
The net consumption of household $i$ and of the community are
\begin{equation}
	z_{i,t} := c_{i,t}+p_{i,t}-r_{i,t},
	\qquad
	z_t := \textstyle\sum_{i\in[N]} z_{i,t}.
	\label{eq:net-consumption}
\end{equation}
Household $i$ is net-consuming if $z_{i,t}>0$ and net-producing if
$z_{i,t}<0$; the community is defined analogously on $z_t$. 

\subsubsection{Community cost under NEM}
The community operates under the utility's NEM tariff. The NEM tariff is parameterized by $\bm{\pi}:=(\pi^+,\pi^-)$ with $\pi^+>\pi^-$:
\begin{equation}\label{eq:NEM_payment}
	P_{\bm{\pi}}(z_t) = \big[\mathbbm{1}(z_t>0)\,\pi^+ + \mathbbm{1}(z_t\le0)\,\pi^-\big]\,z_t,
\end{equation}
where $\pi^+$ denotes the retail (import) rate and $\pi^-$ the compensation (export)
rate.

\subsubsection{Community pricing}
Within the community, each household's payment is determined by a uniform,
time-varying payment function applied to its own net consumption $z_{i,t}$. The coordinator sets the price parameter $\bm{\psi}_t=(\psi^+_t,\psi^-_t)$ through a pricing policy $\chi: \bm{x}_t \mapsto \bm{\psi}_t$, and household $i$ pays $P_{\pmb \psi_t}(z_{i,t})$ that has the same form as (\ref{eq:NEM_payment}) with $(\pip, \pim)$ replaced by $(\psi_t^+, \psi_t^-)$. As with the NEM tariff, a household with $z_{i,t}<0$ is credited for its net production.

\subsection{Distributed scheduling as bilevel optimization}\label{sec:bilevel}
We formulate the distributed scheduling problem as a bilevel stochastic optimization: the upper level determines the coordinator's pricing policy $\chi$, and the lower level determines each household's surplus-maximizing
response.

\subsubsection{Lower level, individual household surplus maximization}
For a fixed pricing policy $\chi$, the broadcast prices $\{\pmb \psi_t^{\chi}\}_{t=1}^T$ form a stochastic process induced by $\chi$ and the community state dynamics. Households, as price takers, treat $\{\pmb \psi_t^{\chi}\}$ as exogenous. 

Within the household, the two demands are scheduled sequentially in each interval: the TCL, being non-deferrable, is scheduled first with priority claim on the household's renewable generation; the EV charging decision is then made against the {\em residual} renewable. Given the broadcast price parameter $\pmb \psi_t^\chi $ and the realized renewable $r_{i,t}$, household $i$ sets its TCL consumption by the surplus maximization
\begin{equation}\label{eq:tcl-opt}
	p_{i,t}(\pmb \psi_t^\chi )\in\argmax_{p\,\in\,\mathbb{R}_+}\;
	U_{i,t}(p)-P_{\pmb \psi_t^\chi}(p-r_{i,t}).
\end{equation}
The TCL is scheduled as if it were the household's only load, and the EV has a claim only on what remains. 

With the TCL scheduled by \eqref{eq:tcl-opt}, the household's remaining decision is the EV charging policy $\pmb \mu_i$, chosen to maximize the expected cumulative surplus
\begin{equation}\label{eq:lower}
	\begin{aligned}
		S_i^{\chi}:=\max_{\pmb \mu_i}\;
		&\mathbb{E}\Bigg[\sum_{t=1}^{T}
		 U_{i,t}\big(p_{i,t}(\pmb \psi_t^{\chi})\big)
		- P_{\pmb \psi_t^{\chi}}(z_{i,t})\\[-2pt]
		&\qquad\qquad -\mathbbm{1}(\tau_{i,t}=1)\,q(d_{i,t}-c_{i,t})\Bigg]\\
		\text{s.t.}\;\;
		& (1)-(3), (5)\;\;
		c_{i,t}=\mu_{i,t}(\pmb x_{i,t},\pmb \psi_t^{\chi}),\;\forall t.
	\end{aligned}
\end{equation}
\subsubsection{Pricing constraints}
The coordinator's pricing policy is subject to two constraints, ensuring coalition stability and no deficit for the coordinator. 

\begin{definition}[Revenue Adequacy]\label{def:ra}
	A pricing policy $\chi$ is {\em revenue adequate} if, 
	\begin{equation}\label{eq:ra}
		P_{\pmb \pi}\Big(\sum_{i\in[N]}z_{i,t}^{\chi}\Big)
		\le\sum_{i\in[N]}P_{\pmb \psi_t^\chi }\big(z_{i,t}^{\chi}\big),\quad\forall t,
	\end{equation}
	where $z_{i,t}^{\chi}$ is the $\chi$-induced net consumption of
	household $i$.
\end{definition}

\begin{definition}[Individual Rationality]\label{def:ir}
	A pricing policy $\chi$ is {\em individually rational} if 
	\begin{equation}\label{eq:ir}
		S_i^{\chi}\ge S_i^{\text{\normalfont NEM}},\quad\forall i\in[N],
	\end{equation}
	with $S_i^{\text{\normalfont NEM}}$ defined as \eqref{eq:lower} with $P_{\pmb \psi_t^{\chi}}$ and $p_{i,t}(\pmb \psi_t^\chi)$ replaced by $P_{\pmb \pi}$ and $p_{i,t}(\pmb \pi)$. 
\end{definition}


\subsubsection{Upper level, community surplus maximization}
The community surplus in interval $t$, $W_t$, is defined as
\[
W_t(\pmb x_t, \pmb p_t, \pmb c_t) := \sum_{i\in[N]}U_{i,t}\big(p_{i,t}\big)
- P_{\pmb \pi}(z_t)-\sum_{j\in\mathcal{J}_t}q(d_{j,t}-c_{j,t}),
\]
where $\mathcal{J}_t:=\{j\in[N]\,|\,\tau_{j,t}=1\}$ is the set of EVs reaching their deadlines in interval $t$. The coordinator's problem is the total community surplus maximization:
\begin{equation}\label{eq:upper}
	\begin{aligned}
		\max_{\chi}\;\;
		&\mathbb{E}\Bigg[\sum_{t=1}^{T}
		W_t(\pmb x_t, \pmb p_t, \pmb c_t)
		\Bigg]\\
		\text{s.t.}\;\;
		&(1)-(5), (7)-(8),\\
		&c_{i,t}= \mu_{i,t}^{\chi}(\pmb x_{i,t}, \pmb \psi_t),\;\;
		p_{i,t}=p_{i,t}(\pmb \psi_t),\;\forall i,t,
	\end{aligned}
\end{equation}
where $\pmb \mu_{i}^{\chi}$ is the optimal lower-level policy of \eqref{eq:lower}. We assume the coordinator knows the deadlines and remaining demands of all plugged-in EVs and can measure each household's renewable generation.

\section{Optimal Centralized Scheduling}\label{sec:censolution}
We first consider the centralized problem, with the coordinator observing every household state. The coordinator solves the welfare maximization \eqref{eq:upper} for the scheduling policy $\pmb \nu = (\nu_1, \ldots, \nu_T)$ without pricing constraints:
\begin{equation}\label{eq:centralized}
	\begin{aligned}
		W^N(\pmb x_1):=\max_{\pmb \nu}\;
		&\;\mathbb{E}\Bigg[\sum_{t=1}^{T}W_t(\pmb x_t, \pmb p_t, \pmb c_t)\Bigg]\\
		\text{s.t.}\;\;
		&\;(1)-(4),\;\;
		(\pmb p_t,\pmb c_t)=\nu_t(\pmb x_t),\;\forall t,
	\end{aligned}
\end{equation}
a finite-horizon Markov decision process with state $\pmb x_t$ and action as $N$ charging and TCL decisions. This problem serves as an upper bound of the welfare achievable by any pricing policy. We characterize the structure of the optimal policy here. 

 
\subsection{Optimal centralized policy}\label{sec:opt-structure}
Two state-dependent quantities bracket each EV charging decision. To
meet the deadline, the EV $i$ must charge at least
\begin{equation*}\label{eq:m}
	m_{i,t}(d_{i,t},\tau_{i,t}):=\max\{d_{i,t}-(\tau_{i,t}-1)\bar c,\,0\}
\end{equation*}
in interval $t$, and the charger can deliver at
most
\begin{equation*}\label{eq:M}
	M_{i,t}(d_{i,t}):=\min\{d_{i,t},\bar c\}.
\end{equation*}

For the non-deferrable demand, writing $\partial U_{i,t}^{-1}$ for the inverse marginal utility, define
\begin{equation*}\label{eq:tcl-levels}
	p_{i,t}^+:=\partial U_{i,t}^{-1}(\pi^+),\qquad
	p_{i,t}^-:=\partial U_{i,t}^{-1}(\pi^-),
\end{equation*}
the solution of (\ref{eq:tcl-opt}) for the constant prices $\pip$ and $\pim$. Concavity of $U_{i,t}$ gives $0 < p_{i,t}^+\le p_{i,t}^-$, where positivity is assumed of $U_{i,t}$. We now formalize two-threshold structure of the optimal centralized policy. 

\begin{theorem}[Optimal centralized policy]\label{thm:centralized}
	The optimal policy $\nu^*_t: \pmb x_t\mapsto (c^*_t,p^*_t)$ of
	\eqref{eq:centralized} is a two-threshold policy on the aggregate
	renewable $r_t$, with thresholds
	$\Delta_t^+ := \sum_{i\in[N]}p^+_{i,t}+m_{i,t}(d_{i,t}, \tau_{i,t})$ and
	$\Delta_t ^- := \sum_{i\in[N]}p^-_{i,t}+M_{i,t}(d_{i,t})$, and a net-zero-zone policy
	$\pmb \rho_t(\pmb x_t)=(\rho_{1,t}(\pmb x_t),\ldots,\rho_{N,t}(\pmb x_t))$:
	\begin{equation*}\label{eq:thm}
		(p^*_{i,t},c^*_{i,t})=
		\begin{cases}
			\big(p^+_{i,t},\,m_{i,t} (d_{i,t}, \tau_{i,t})\big), &
			r_t\le \Delta_t^+,\\[3pt]
			\rho_{i,t}(\pmb x_t), &
			\Delta_t^+ < r_t \le \Delta_t^- \\
			\big(p^-_{i,t},\,M_{i,t}(d_{i,t})\big), &
			\Delta_t^- < r_t
		\end{cases}
	\end{equation*}
	where $\rho_{i,t}(\pmb x_t) = (p_{i,t}^\rho, c_{i,t}^{\rho})$ satisfies $\sum_{i\in[N]}\big( p^{\rho}_{i,t} + c^{\rho}_{i,t}\big)=r_t$.
\end{theorem}

The thresholds partition $r_t$ into a net-consuming zone below $\Delta_t^+$, a net-zero zone between the thresholds, and a net-producing zone above $\Delta_t^-$. The thresholds are themselves the aggregate consumption in the two outer zones, so the community imports the shortfall below $\Delta_t^+$, exports the surplus above $\Delta_t^-$. 

Both decisions in the outer zones follow from the price gap. Below $\Delta_t^+$, the marginal kWh is bought at $\pip$, so the TCL consumes only up to $p_{i,t}^+$ and each EV charges only what its deadline requires, deferring the rest for future renewables. Above $\Delta_t^-$, the marginal kWh is exported at $\pim$, so both loads expand: the TCL to $p_{i,t}^-$ and the EV to $M_{i,t}$, since consuming is preferable to selling at the lower rate. Note that, in the net-consuming and net-producing zones, the optimal decisions and thresholds are in closed form, with computational cost linear in $N$. The complexity is isolated to the net-zero zone, where the optimal allocation $\pmb \rho_t$ requires solving the Bellman equation, which suffers from the curse of dimensionality.

Furthermore, $(p^+_{i,t},m_{i,t})$ and $(p^-_{i,t},M_{i,t})$ are household $i$'s best responses to $\pi^+$ and $\pi^-$, respectively, so outside the net-zero zone, the centralized optimum can be induced by broadcasting the right price. Sec.~\ref{sec:tpr} builds on this fact to construct a pricing rule for distributed scheduling.

\section{Threshold Pricing Rule and Properties}\label{sec:tpr}
Rather than solving the bilevel program directly, we adopt the three-zone structure of Theorem~\ref{thm:centralized} and seek prices that induce the centralized decisions as individual best responses. Such prices exist in the two outer zones, but not in the net-zero zone, where we trade optimality for simplicity by posting the NEM tariff.

\subsection{Threshold Pricing Rule (TPR)}\label{sec:tpr-def}
At the beginning of each interval, we assume that household $i$ reports its EV state $(d_{i,t},\tau_{i,t})$ and its price-contingent TCL levels $(p^+_{i,t},p^-_{i,t})$. Utility function $U_{i,t}$ is still private. Given the reports and the measured renewables, the coordinator sets the community price by the following rule.

\begin{definition}[Threshold Pricing Rule]\label{def:tpr}
	The TPR pricing policy $\chi_{\textrm{\upshape TPR}}:\pmb x_t\mapsto \pmb \psi_t$ sets the price according to the two thresholds of the optimal centralized policy $\Delta_t^+ $ and $\Delta_t^-$:
	\begin{equation*}\label{eq:tpr}
		\pmb \psi_t=
		\begin{cases}
			\pmb\pi^+:=(\pi^+,\pi^+), &
			r_t\le \Delta_t^+,\\[2pt]
			\pmb\pi:=(\pi^+,\pi^-), &
			\Delta_t^+<r_t\le \Delta_t^-\\[-1pt]
			\pmb\pi^-:=(\pi^-,\pi^-), &
			r_t> \Delta_t^-.
		\end{cases}
	\end{equation*}
\end{definition}
TPR broadcasts the linear price $\pmb \pi^+$ in the net-consuming zone, the linear price $\pmb \pi^-$ in the net-producing zone, and the NEM tariff $\pmb \pi$ in between. Computating prices requires evaluating the two thresholds, which are the sum of $N$ reported terms. 

\subsection{Optimal household response to TPR}\label{sec:tpr-response}

The following proposition solves the household's surplus maximization \eqref{eq:lower} under TPR.

\begin{proposition}[Optimal household response]\label{prop:response}
	Under $\chi_{\textrm{\upshape TPR}}$, household $i$'s optimal decision
	is myopic, depending only on the current broadcast price:
	\begin{equation*}\label{eq:response}
		(p_{i,t},c_{i,t})=
		\begin{cases}
			\big(p^+_{i,t},\,m_{i,t} (d_{i,t}, \tau_{i,t})\big), & \pmb \psi_t=\pmb\pi^+,\\[2pt]
			\big(\tilde p_{i,t}(\pmb\pi, r_{i,t}),\,\tilde c_{i,t}(d_{i,t}, \tau_{i,t})\big), & \pmb \psi_t=\pmb\pi,\\[2pt]
			\big(p^-_{i,t},\,M_{i,t} (d_{i,t})\big), & \pmb \psi_t=\pmb\pi^-,
		\end{cases}
	\end{equation*}
	with $\tilde p_{i,t}(\pmb \pi, r_{i,t})$ and $\tilde c_{i,t}(d_{i,t}, \tau_{i,t})$ are defined as
	\begin{equation}\label{eq:tcl-nem}
		\begin{aligned}
			\tilde p_{i,t}(\pmb \pi, r_{i,t}) = 
			\min\big\{
			\max \{r_{i,t}, p_{i,t}^+\}, p_{i,t}^-
			\big\}
		\end{aligned}
	\end{equation}
	\begin{equation}\label{eq:ev-nem}
		\begin{aligned}
			\tilde c_{i,t}(d_{i,t}, \tau_{i,t})=
			\min\big\{&\max\{r_{i,t}- \tilde  p_{i,t}(\pmb\pi, r_{i,t}),\\
			&m_{i,t}(d_{i,t}, \tau_{i,t})\}, \, M_{i,t}(d_{i,t})\big\}.
		\end{aligned}
	\end{equation}
\end{proposition}

Under $\pmb\pi^+$, every kWh trades at $\pi^+$: the TCL consumes up to $p^+_{i,t}$, each EV charges only its minimum $m_{i,t}$, and households with surplus renewables sell at the high rate $\pi^+$ rather than charge ahead of need. Under $\pmb\pi^-$, every kWh trades at $\pi^-$, so the TCL consumes up to $p^-_{i,t}$ and each EV charges to $M_{i,t}$ rather than sell at the low rate. Under $\pmb\pi$, the household faces the NEM tariff and behave exactly as a stand-alone customer, the TCL absorbing the local generation within $[p^+_{i,t},p^-_{i,t}]$, and the EV claiming the residual. 


\subsection{Properties of TPR}\label{sec:tpr-feasibility}

Here, we establish that the TPR satisfies the two pricing constraints of the upper-level optimization.

\begin{proposition}[Individual rationality and revenue adequacy \cite{jeon2026optimal}]\label{prop:ir}
	Under $\chi_{\textrm{\upshape TPR}}$, $S_i^{\chi_{\textrm{\upshape TPR}}}\ge S_i^{\textrm{\upshape NEM}}$ for all $i \in [N]$, and $\chi_{\textrm{\upshape TPR}}$ is revenue adequate in every interval.
\end{proposition}
Individual rationality holds because TPR offers a price at least as favorable as NEM in every zone: in the net-consuming zone, a net-exporting household sells at $\pip$ rather than $\pim$, and in the net-producing zone, a net-importing household buys at $\pim$ rather than $\pip$. Revenue adequacy holds because the internal price is never more generous than the price the coordinator faces from the utility: the outer zones are revenue-neutral, and in the net-zero zone the NEM price gap guarantees a non-negative surplus.

TPR is suboptimal only in net-zero zone. Under a homogeneous community---Bernoulli arrival at rate $\alpha$, durations bounded by $\bar T$, renewables i.i.d. across households and stationary with mean $\theta_r$, and common utility $U$ so that $p_{i,t}^+ = p^+$ and $p_{i,t}^- = p ^-$---that zone is visited with vanishing probability as $N$ grows. 
\begin{theorem}[Asymptotic optimality of TPR \cite{jeon2026optimal}]\label{thm:asymptotic}
	Let $W^{N,\chi_{\textrm{\upshape TPR}}}(\pmb x_1)$ be the expected
	community welfare under pricing rule $\chi_{\textrm{\upshape TPR}}$. If
	$\alpha<(\theta_r - p^-)/(\bar c\,\bar T)$, then
	\begin{equation*}\label{eq:asymptotic}
		\lim_{N\to\infty}\frac{1}{N}
		\big(W^N(\pmb x_1)-W^{N,\chi_{\textrm{\upshape TPR}}}(\pmb x_1)\big)=0.
	\end{equation*}
\end{theorem}
The light-traffic condition $(\alpha<(\theta_r - p^-)/(\bar c\,\bar T))$ requires the arrival rate to be low enough so that aggregate generation can serve every TCL at $p^-$ and can fully charge every EV present. The quantity $N(\theta_r - p^-)/\bar c$ represents the number of EVs that the community's surplus generation can serve simultaneously, and the bound tightens with $\bar T$ because longer durations keep more EVs present at once. 
\section{Numerical Simulations} \label{sec:numericalsimulation}
This section evaluates the proposed distributed scheduling algorithm numerically; we demonstrate asymptotic optimality and individual surplus gains. 

\subsection{Simulation setting }\label{sec:simulation_setting}
Each of the $N$ community members had a BTM solar generator, an EV charger, and a TCL, scheduled over a 24-hour horizon with 1-hour intervals. 

\subsubsection{Evaluation datasets}
Experiments combine real and synthetic data. EV demand data are drawn from Caltech's Adaptive Charging Network (ACN-Data) dataset \cite{lee_acndata_2019}. Renewable generation is represented two different ways depending on the experiment: i.i.d. draws from a stationary lognormal distribution for the asymptotic-optimality study, and measured BTM solar output from Pecan Street's New York residential customers for the individual-surplus-gain study.



\subsubsection{EV and TCL demands}
The arrival of EVs was modeled as a homogeneous Bernoulli process, with each vehicle's charging duration drawn from a truncated Gaussian centered at 5 hours. Charging was capped at 7.2 kWh. 

We modeled the utility of TCL demand as a time-invariant quadratic function, $U_i(p_{i,t}) = a_i p_{i,t} - \frac 1 2 b_i p_{i,t}^2$. For the asymptotic optimality simulation, we set $a_i = 1$ and $b_i = 1.5$ for all $i$. For the individual surplus gain simulation, we estimated $a_i$ and $b_i$ from Pecan Street household consumption profiles using the elasticity value in \cite{asadinejad2018evaluation} and the NEM rates.  

\subsubsection{Energy price}
The time-invariant NEM rates were $(\pi^+, \pi^-) = (0.5$ \$/kWh, $0.2$ \$/kWh$)$.

\subsection{Asymptotic optimality of distributed scheduling} 
We validate the asymptotic optimality of TPR established in Theorem~\ref{thm:asymptotic}. Since computing the centralized optimal policy is intractable, we compared TPR against the Oracle policy that solves an open-loop optimization with perfect knowledge of all random variables, providing an upper bound on achievable surplus. We also compared two suboptimal centralized scheduling policies: the threshold policy in Theorem~\ref{thm:centralized} with TCL scheduled by (\ref{eq:tcl-nem}) and remaining aggregate renewables allocated to plugged-in EVs by least-laxity-first (LLF) for the net-zero zone policy $\pmb \rho_t$, and model predictive control (MPC).
\setlength{\belowcaptionskip}{-8pt}
 \begin{figure}[t]
 	\centering
 	\includegraphics[width=0.7\columnwidth]{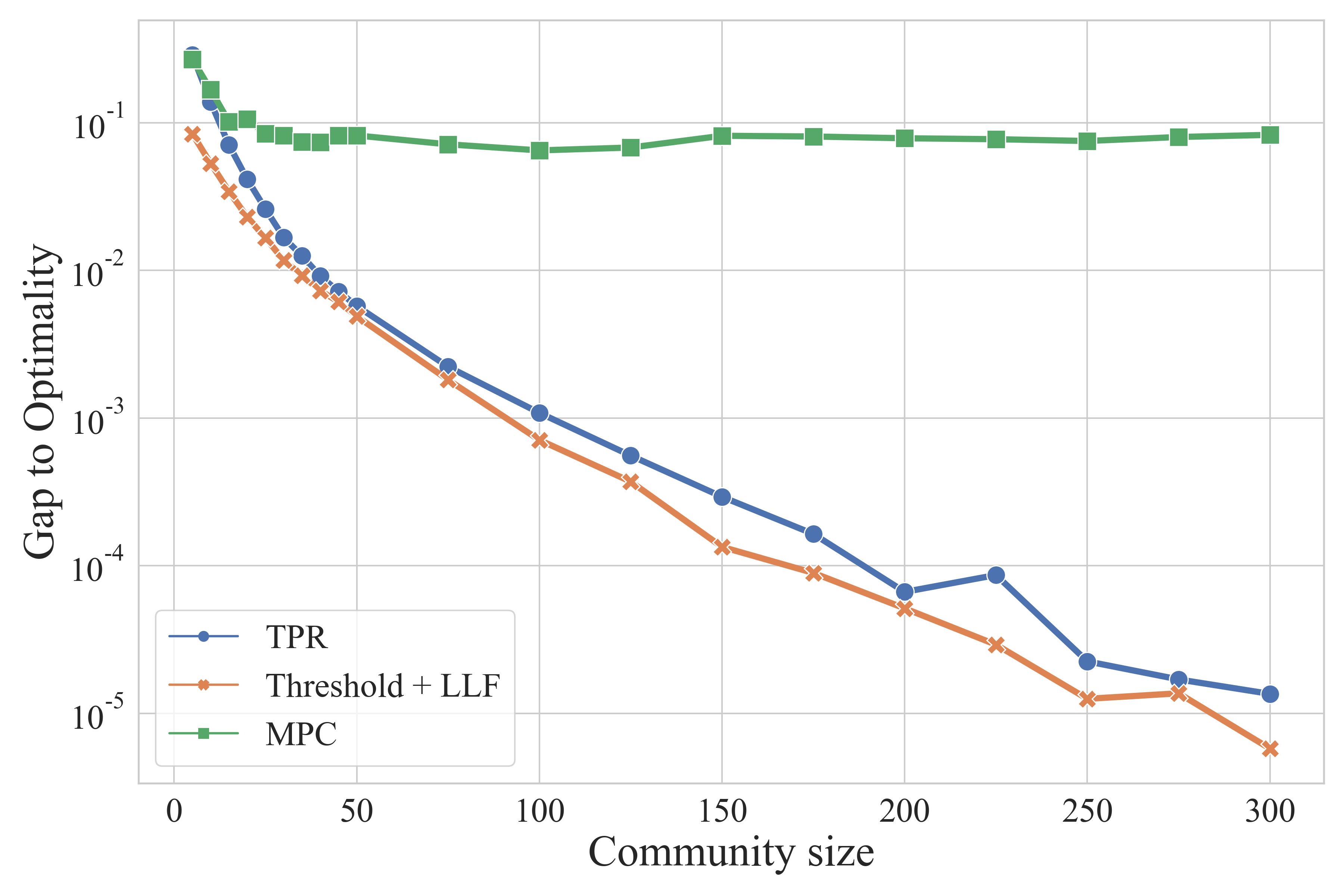}
 	\caption{Optimality gap of centralized and distributed policies (log $y$ axis). Renewable generation has mean $\theta_r = 2$ kWh, and EV arrival rate is $\alpha = (\theta_r - p^-) / (\bar c \bar T)$, with $\bar T = 6$ hours.}
 	\label{fig:asymptoticopt}
 \end{figure}
 
Fig~\ref{fig:asymptoticopt} plots, on a log scale, the per-household gap in community surplus between each policy and the Oracle benchmark. Both TPR and the threshold policy with the LLF net-zero zone scheduling exhibited exponential decay as the community grew, with the latter converging faster, while the gap of MPC persisted. Although Theorem~\ref{thm:asymptotic} only establishes asymptotic optimality only for TPR, suggests that asymptotic optimality arises from the threshold structure shared by TPR and the centralized policy, and that convergence may in fact be exponential. 

\begin{figure}[t]
	\centering
	\includegraphics[width=0.7\columnwidth]{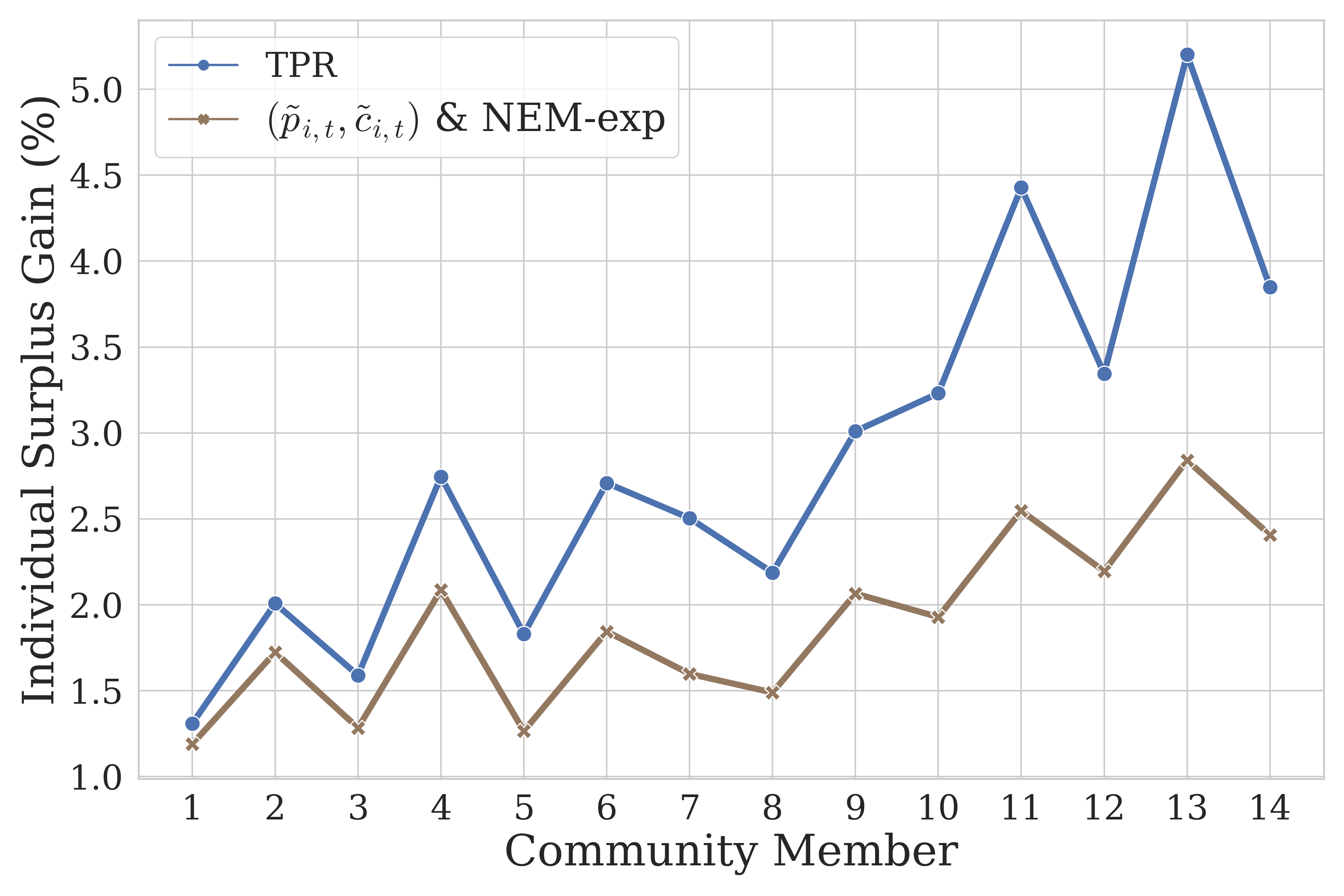}
	\caption{Individual surplus gains in percentage}
	\label{fig:indcostsaving}
\end{figure}

\subsection{Comparisons of individual surplus gains}
We validate TPR's individual rationality by measuring each member's surplus gain under TPR as a percentage of its stand-alone surplus under NEM. We also compare the individual surplus gain to that under the pricing rule proposed in \cite{chakraborty2018analysis}, referred to as NEM \textit{ex-post} (NEM-exp). NEM-exp decides community prices based on the sign of the community net-consumption $z_t$:  
\begin{equation}\label{eq:chak_pricingrule}
	P_{\mbox{\tiny NEM-exp}} (z_{i,t}) = \begin{cases}
		\pi^+ z_{i,t}, & z_t > 0 \\
		\pi^- z_{i,t}, & z_t \le 0
	\end{cases}.
\end{equation}
Since price is decided only after $z_{i,t}$ is realized, it requires a separate scheduling policy. We use the optimal decision for a stand-alone customer under NEM, which is $(\tilde p_{i,t}, \tilde c_{i,t})$ in equations (\ref{eq:tcl-nem}) and (\ref{eq:ev-nem}). 

In the community of size $N = 14$, we indexed members by descending order of renewable generation capacity. In Fig.~\ref{fig:indcostsaving}, TPR delivered 10.08\%–83.13\% larger surplus gains than $(\tilde p_{i,t}, \tilde c_{i,t})$ with (\ref{eq:chak_pricingrule}). Households with smaller generation capacity achieved the larger gains, benefiting more from purchasing shared renewables at price $\pim$ to charge their EVs and schedule TCLs with higher consumption surplus.

\section{Conclusions}
This paper developed a pricing mechanism for the joint distributed scheduling of deferrable and non-deferrable flexible demands, derived from the structure of the optimal centralized schedule and shown to be individually rational, revenue adequate, and asymptotically optimal in community welfare.

Two directions remain open. TPR's thresholds are computed from the consumption levels members report, so a member could shift a zone boundary in its favor by misreporting; designing a rule under which truthful reporting is incentive compatible is a natural next step. The suboptimality lies entirely in the net-zero zone, where the optimal allocation requires solving a Bellman equation; characterizing the welfare loss of TPR in the net-zero zone would provide a suboptimality bound for finite communities. 

\bibliographystyle{IEEEtran}
\bibliography{JeonTongZhaoRef.bib}
\appendices
\section{Proof of Theorem 1 and Proposition 1}
\subsection{Preliminaries and notation}

Write $\pmb d_t = (d_{i,t})_{i \in [N]}$, $\pmb \tau_t = (\tau_{i,t})_{i \in [N]}$. The feasible action set at $t$ is
\[
\mathcal F_t(\pmb d_t) := 
\{
(\pmb p, \pmb c)\in \mathbb R_+^N \times \mathbb R_+^N:
c_i \le M_{i,t} (d_{i,t}).
\}
\]
Define the total consumption as $y(\pmb p, \pmb c):= \sum_i  (p_i + c_i)$, so the community net consumption is $z_t = y(\pmb p, \pmb c) - r_t$. Problem (\ref{eq:centralized}) is a finite-horizon MDP whose value function satisfies $V_{T+1} \equiv 0$, and for $t \in \mathcal T$, 
\begin{align}
	&V_t(\pmb x_t) = \underset{(\pmb p, \pmb c)\in \mathcal F_t(\pmb d_t)}{\max} J_t(\pmb p, \pmb c), \label{eq:Bellman}\\
	&J_t(\pmb p, \pmb c) := \sum_{i\in [N]} U_{i,t}(p_i) - P_{\pmb \pi}(y(\pmb p, \pmb c) - r_t ) + \Phi_t(\pmb x_t, \pmb c) \label{eq:J_t}\\
	&\Phi_t(\pmb x_t, \pmb c) := 
	-\sum_{j \in \mathcal J_t} q(d_{j,t} - c_j) + \mathbb E[V_{t+1}(\pmb x_{t+1}) | \pmb x_t, \pmb c].\label{eq:Phi_t}
\end{align}
By (\ref{eq:Bellman})-(\ref{eq:J_t}), $\pmb x_{t+1}$ depends on $\pmb c$ only through the residual demands $d_{i,t} - c_{i,t}$ of the chargers with $\tau_{i,t} \ge 2$; for $i \in \mathcal J_t$ the vehicle departs and the charger is idle at $t+1$ regardless of $c_{i,t}$. Call coordinate $i$ \textit{feasible} at $t$ if $d_{i,t} \le \tau_{i,t} \bar c$, so that the deadline can still be met; by footnote 1, every arrival is feasible. 

Since $\pip > \pim$, the NEM payment (\ref{eq:NEM_payment}) is the upper envelope of two linear prices, 
\begin{equation}\label{eq:NEM_maxform}
P_{\pmb \pi}(z) = \max\{\pip z, \pim z\}.
\end{equation}
Accordingly define, for $\pi^\bullet\in\{\pi^+,\pi^-\}$, the \emph{linear-price} objectives
\[
J_t^{\bullet}(\bp, \bc):= \sum_{i\in [N]}U_{i,t}(p_i )-
\pi^{\bullet}(y(\bp, \bc) - r_t) + \Phi_t(\bx_t, \bc),
\]
so that 
\[
J_t = \min\{J_t^+, J_t^-\}.
\]

\subsection{Sandwich bound on the value of deferrable demand}
The entire proof rests on the following bound: one kWh of EV demand can always be served later at a marginal cost at most $\pip$, and relieving one kWh saves at least $\pim$. 

\begin{lemma}\label{lemma1}
	For every $t$ and every $(\btau_t, \br_t)$, the mapping $\bd \mapsto V_t(\bd, \btau_t, \br_t)$ satisfies
	\begin{enumerate}
		\item it is concave on $\mathbb R_+^N$;
		\item $V_t(\bd) - V_t(\bd') \ge \pim \mathbf 1^T(\bd' - \bd)$ for all $\bd \le \bd'$;
		\item $V_t(\bd) - V_t(\bd') \le \pip \mathbf 1^T(\bd' - \bd)$ for all $\bd \le \bd'$ with $\bd'$ feasible;
		\item $V_t(\bd) - V_t(\bd') \ge q'(0)(d_i' - d_i)$ whenever $\bd, \bd'$ differ only in coordinate $i$ and $d_i' \ge d_i \ge \tau_{i,t} \bar c$, 
	\end{enumerate}
	where $\bd \le \bd'$ is defined as $d_i \le d_i'$ for all $i$. 
\end{lemma}

\begin{proof}
	Backward induction on $t$, the claims being vacuous at $t = T+1$. Assume they hold at $t+1$. 
	
	\begin{enumerate}
		\item The stage reward in (\ref{eq:Bellman}) is jointly concave in $(\bd, \bp, \bc)$: each $U_{i,t}$ is concave, $-P_{\pmb \pi}(\cdot)$ is concave, and $-q(d_j - c_j)$ is concave. The continuation term is concave in $(\bd, \bc)$ by the induction hypothesis, since the successor demand $d_{i,t} - c_{i,t}$ is affine and the expectation preserves concavity. The constraint set $\{(\bd, \bp, \bc): \bp \ge 0, 0 \le c_i \le d_i, c_i \le \bar c\}$ is convex. Partial maximization of a jointly concave function over a convex set is concave in the remaining varaible $\bd$. \\
		
		\item Fix $\bd \le \bd'$ and let $(\bp, \bc)$ be optimal at $\bd'$. Set $c_i^{\circ}:= \min\{c_i, M_{i,t}(d_i)\}$ and $\delta_i := c_i - c_i^\circ \in [0, d_i' - d_i]$, so $(\bp, \bc^\circ)$ is feasible at $\bd$ and 
		\[
		(d_i - c_i^\circ) - (d_i' - c_i) = \delta_i - (d_i' - d_i) \le 0.
		\]
		Relative to $(\bp, \bc)$, the action $(\bp, \bc^\circ)$ consumes $\mathbf 1^T \pmb \delta $ less energy, which by (\ref{eq:NEM_maxform}) lowers the payment by at least $\pim \mathbf 1^T \pmb \delta$, and it leaves every residual demand smaller by $(d_i' - d_i) - \delta_i$, which by the induction hypothesis raises $\Phi_t $ by at least $\pim \sum_i [(d_i' - d_i) - \delta_i]$. Adding the two gives $V_t(\bd) \ge J_t(\bp, \bc^\circ) \ge V_t(\bd' ) + \pim \mathbf 1^T (\bd' - \bd)$. \\
		
		\item  Fix $\bd \le \bd'$ with $\bd'$ feasible and let $(\bp, \bc)$ be optimal at $\bd$. Put $c_i^\circ:= c_i + \min\{d_i' - d_i, \bar c - c_i\}$ and $\eta_i := c_i^\circ - c_i$, so $(\bp, \bc^\circ)$ is feasible at $\bd'$. It consumes $\mathbf 1^T \pmb \eta$ more energy, raising the payment by at most $\pip \mathbf 1^T \pmb \eta$, and its residual demands exceed those of $(\bp, \bc)$ by $(d_i' - d_i) - \eta_i$; each such residual is at most $d_i' - c_i^\circ \le \tau_{i,t} \bar c- \bar c = (\tau_{i,t} - 1)\bar c$ when $\eta_i \le d_i' - d_i$ (the cap binds), hence feasible at $t+1$, and equals $d_i - c_i$ otherwise. The induction hypothesis (or convexity of $q$ for $i \in \mathcal J_t$, whose successor is unaffected) therefore bounds the loss in $\Phi_t$ by $\pip \sum_i [(d_i' - d_i) - \eta_i]$. Adding, $V_t(\pmb d') \ge J_t(\bp, \bc ^\circ) \ge V_t(\bd) - \pip \mathbf 1^T (\bd' - \bd)$. \\
		
		\item Let $d_i' \ge d_i \ge \tau_{i,t} \bar c$ and let $(\bp, \bc)$ be optimal at $\bd'$. Since $d_i' \ge \bar c$, the action $(\bp, \bc)$ is feasible at $\bd$ as well, and it costs the same there; only the residual of coordinate $i$ chages, from $d_i' - c_i$ down to $d_i - c_i \ge \tau_{i,t} \bar c- \bar c = (\tau_{i,t} - 1)\bar c$. If $\tau_{i,t} \ge 2$, the induction hypothesis gives a gain of at least $q'(0)(d_i' - d_i)$. If $\tau_{i,t} = 1$ the gain is $q(d_i' - c_i) - q(d_i - c_i)\ge q'(0) (d_i' - d_i)$ by convexity of $q$ and $d_i - c_i \ge 0$. In either case, $V_t(\bd )\ge V_t(\bd') + q'(0) (d_i' - d_i)$. 
	\end{enumerate}
\end{proof}
Let $\pmb m_t := (m_{i,t}(d_{i,t}, \tau_{i,t}))_i$, $\pmb M_t := (M_{i,t}(d_{i,t}))_i$ and $\mathcal C_t := \prod_{i\in[N]} [m_{i,t}, M_{i,t}].$

\begin{corollary}
	Assume every coordinate is feasible at $t$. Then $\Phi_t(\bx_t, \cdot)$ is concave, every $\bc \in \mathcal F_t(\bd)$ with $\bc \notin \mathcal C_t$ is strictly dominated, and for every $\bc \in \mathcal C_t$, 
	\begin{align}
		\pim \mathbf 1^T (\pmb M_t - \bc) \le \Phi_t(\bx_t, \pmb M_t) - \Phi_t(\bx_t, \bc) \label{eq:18}\\
		\Phi_t(\bx_t, \bc ) - \Phi_t(\bx_t, \pmb m_t) \le \pip \mathbf 1^T (\bc - \pmb m_t)
		\label{eq:19}.
	\end{align}
\end{corollary}

\begin{proof}
	Concavity was shown in the proof of Lemma~\ref{lemma1} (1). For $\bc \in \mathcal C_t$ the residual demands satisfy $d_{i,t} - c_{i,t} \le d_{i,t} - m_{i,t} \le (\tau_{i,t} - 1) \bar c$, so all successor states are feasible; equations (\ref{eq:18}) and (\ref{eq:19}) are then Lemma~\ref{lemma1} (2) applied between the residuals of $\bc$ and $\pmb M_t$, and Lemma~\ref{lemma1} (3) applied between those of $\pmb m_t$ and $\pmb c$ (for $i \in \mathcal J_t$ we have $m_{i,t} = d_{i,t} = M_{i,t}$ by feasibility, so those coordinates are inactive). 
	
	Suppose $c_i < m_{i,t}$ for some $i$; then $m_{i,t} > 0$, hence $d_{i,t} - m_{i,t} = (\tau_{i,t} - 1) \bar c$ and $d_{i,t} - c_i > (\tau_{i,t} - 1) \bar c$. Raising $c_i$ to $m_{i,t}$ increases the payment by at most $\pip(m_{i,t} - c_i)$ by (\ref{eq:NEM_maxform}), and increases $\Phi_t$ by at least $q'(0)(m_{i,t} - c_{i,t})$ by Lemma~\ref{lemma1}(4). Since $q'(0) > \pip$, the move strictly improves $J_t$. 
\end{proof}

\subsection{Proof of Theorem~\ref{thm:centralized}}

Fix $t$ and $\bx_t$; by footnote 1 and Corollary 1, every coordinate is feasible and we may restrict to maximization in (\ref{eq:Bellman}) to $\mathbb R_+^N \times \mathcal C_t$. 

\emph{Step 1: the two linear-price problems.} Let
$\bp_t^\pm:=(p_{i,t}^\pm)_i$ with $p_{i,t}^\pm=\partial U_{i,t}^{-1}(\pi^\pm)$, and set
$v^+:=(\bp_t^+,\bm m_t)$, $v^-:=(\bp_t^-,\bm M_t)$. We claim
\begin{equation}\label{eq:linmax}
	v^+\in\arg\max_{\mathbb R_+^N\times\mathcal C_t} J_t^{+},\qquad v^-\in\arg\max_{\mathbb R_+^N\times\mathcal C_t}J_t^{-}.
\end{equation}
Indeed, $J_t^{+}$ separates as $\sum_i\big[U_{i,t}(p_i)-\pi^+p_i\big]+\big[\Phi_t(x_t,\bc)-\pi^+\mathbf 1^T\bc\big]+\pi^+r_t$. Each $p_i$-term is concave and maximized at $p_{i,t}^+$ by the first-order condition $\partial U_{i,t}(p_i)=\pi^+$, which has the interior solution $p_{i,t}^+>0$ by concavity and positivity of $\partial U_{i,t}$. The bracketed $\bc$-term is maximized at $\bm m_t$ over $\mathcal C_t$ by the inequality \eqref{eq:19}. The argument for $J_t^-$ is identical, using the inequality  \eqref{eq:18} to place the maximizer at $\bm M_t$. Note that $y(v^+)=\Delta_t^+$ and $y(v^-)=\Delta_t^-$, and that $\Delta_t^+\le\Delta_t^-$ because $p_{i,t}^+\le p_{i,t}^-$ (concavity of $U_{i,t}$ and $\pi^+>\pi^-$) and $m_{i,t}\le M_{i,t}$. \\

\emph{Step 2: net-consuming zone, $r_t\le\Delta_t^+$.} Since $J_t=\min\{J_t^+,J_t^-\}\le J_t^+$ pointwise,
\eqref{eq:linmax} gives $V_t(\bx_t)\le J_t^{+}(v^+)$. At $v^+$ the net consumption is $y(v^+)-r_t=\Delta_t^+-r_t\ge0$,
so $J_t(v^+)=J_t^{+}(v^+)$ by \eqref{eq:NEM_payment}. The upper bound is attained and $v^+$ is optimal. \\

\emph{Step 3: net-producing zone, $r_t>\Delta_t^-$.} Symmetrically, $J_t\le J_t^-$ pointwise gives $V_t(\bx_t)\le J_t^{-}(v^-)$, while $y(v^-)-r_t=\Delta_t^--r_t<0$ implies $J_t(v^-)=J_t^{-}(v^-)$  by \eqref{eq:NEM_payment}. Hence $v^-$ is optimal. \\

\emph{Step 4: net-zero zone, $\Delta_t^+<r_t\le\Delta_t^-$.} Split the feasible set into the convex pieces $A^{+}:=\{(\bp,\bc):y(\bp,\bc)\ge r_t\}$ and $A^{-}:=\{(\bp,\bc):y(\bp,\bc)\le r_t\}$, on which \eqref{eq:NEM_payment} gives $J_t=J_t^{+}$ and $J_t=J_t^{-}$, respectively. Let $\hat v\in A^{+}$ maximize $J_t^{+}$ over $A^+$ and suppose $y(\hat v)>r_t$. Since $y(v^+)=\Delta_t^+<r_t<y(\hat v)$ and $y$ is affine, the segment $[\hat v,v^+]$ contains a point $\tilde v$ with $y(\tilde v)=r_t$, and concavity of $J_t^{+}$ together with $J_t^{+}(v^+)\ge J_t^{+}(\hat v)$ gives $J_t^{+}(\tilde v)\ge J_t^{+}(\hat v)$. Hence $J_t^+$ attains its maximum over $A^{+}$ on the face $\{y=r_t\}$. 

The same argument on $A^{-}$, using $y(v^-)=\Delta_t^-\ge r_t$, places a maximizer of $J_t^-$ over $A^-$ on the same face. As $A^+\cup A^-$ is the whole feasible set and $J_t$ agrees with $J_t^{+}$ and $J_t^{-}$ on the
respective pieces, some maximizer $\rho_t(\bx_t)=(\bp_t^\rho,\bc_t^\rho)$ of $J_t$ satisfies
$\sum_i(p_{i,t}^\rho+c_{i,t}^\rho)=y(\rho_t(\bx_t))=r_t$, i.e.\ the community is net-zero.

Steps 2--4 hold for every $t$ and $\bx_t$, which proves the theorem. 

\subsection{Proof of Proposition~\ref{prop:response}}
\emph{Step 1 (TCL).} Problem (\ref{eq:tcl-opt}) is a static concave program, so its solution is characterized by the sign of the
one-sided derivatives of $U_{i,t}(p)-P_{\bm\psi_t}(p-r_{i,t})$. For $\bm\psi_t=\bm\pi^{\pm}$ the payment is linear and
the first-order condition $\partial U_{i,t}(p)=\pi^{\pm}$ gives $p_{i,t}^{\pm}$, interior by assumption. For
$\bm\psi_t=\bm\pi$ the derivative is $\partial U_{i,t}(p)-\pi^{+}$ on $p>r_{i,t}$ and $\partial U_{i,t}(p)-\pi^{-}$ on
$p<r_{i,t}$, hence negative for $p>p_{i,t}^{-}$ and positive for $p<p_{i,t}^{+}$; since $p_{i,t}^{+}\le p_{i,t}^{-}$ by
concavity of $U_{i,t}$, the maximizer is $\tilde p_{i,t}=\clip\big(r_{i,t},\,p_{i,t}^{+},\,p_{i,t}^{-}\big)$, i.e.\ (\ref{eq:tcl-nem}).

\emph{Step 2 (Reduction).} By the sequential convention of Sec.~II-C-1, the TCL is scheduled first and claims the
renewable with priority, so $\tilde p_{i,t}$ is a deterministic function of $(\bm\psi_t,r_{i,t})$ alone: it does not
depend on $c_{i,t}$ or on the EV state. Define the residual renewable $g_{i,t}:=r_{i,t}-p_{i,t}(\bm\psi_t)$, so that $z_{i,t}=p_{i,t}(\bm\psi_t)+c_{i,t}-r_{i,t}=c_{i,t}-g_{i,t}$.

The price process is exogenous to $i$; hence, on a fixed realization of $\{r_{i,k},\bm\psi_k\}$, the sequence
$\{g_{i,k}\}$ is determined and unaffected by $\pmb \mu_i$, and $\sum_k U_{i,k}(p_{i,k}(\bm\psi_k))$ is a constant. Surplus
maximization (6) therefore reduces to the cost minimization
\begin{equation}\label{eq:red}
	\min_{\mu_i}\ \sum_{k}P_{\bm\psi_k}\big(c_{i,k}-g_{i,k}\big)+\mathbbm{1} (\tau_{i,k}=1)\,q(d_{i,k}-c_{i,k}),
\end{equation}
which is the price-inelastic problem with net renewable $g_{i,k}$. Every $P_{\bm\psi_k}$ is convex piecewise linear with
slopes in $[\pi^-,\pi^+]$, so Lemma \ref{lemma1} applies and
\begin{equation}\label{eq:slope}
	\pi^{-}\epsilon\ \le\ P_{\bm\psi}(z+\epsilon)-P_{\bm\psi}(z)\ \le\ \pi^{+}\epsilon,\qquad \forall z\in\R,\ \epsilon>0.
\end{equation}

\emph{Step 3 (EV, interchange).} By Lemma \ref{lemma1}, restrict to completing sequences, with realized cost
$J(\bm c_i)=\sum_k P_{\bm\psi_k}(c_{i,k}-g_{i,k})$. Path-wise optimality implies optimality in expectation over adapted
policies, and successive jobs decouple, so it suffices to modify any completing $\bm c_i$ so that its $t$-th entry
equals the asserted action, without increasing $J$; repeating at $t+1,t+2,\dots$ gives the claim. Each modification
moves $\epsilon>0$ between $t$ and the earliest feasible interval of the same job window
$\{t+1,\dots,t+\tau_{i,t}-1\}$, which is adapted and preserves completion.

\emph{Case 1} ($\bm\psi_t=\bm\pi^{+}$). Completion forces $c_{i,t}\ge m_{i,t}$. If $c_{i,t}>m_{i,t}$, move $\epsilon$ to
the future: $P_{\bm\pi^{+}}$ has slope $\pi^{+}$ everywhere, so the saving at $t$ is exactly $\pi^{+}\epsilon$ and the
future increase is at most $\pi^{+}\epsilon$ by \eqref{eq:slope}.

\emph{Case 2} ($\bm\psi_t=\bm\pi^{-}$). If $c_{i,t}<M_{i,t}$, completion implies a positive future charge; borrowing
$\epsilon$ from it costs exactly $\pi^{-}\epsilon$ at $t$ and saves at least $\pi^{-}\epsilon$ later.

\emph{Case 3} ($\bm\psi_t=\bm\pi$). Here $\tilde c_{i,t}=\clip\big(g_{i,t},\,m_{i,t},\,M_{i,t}\big)$ by \eqref{eq:ev-nem}. If $c_{i,t}>\tilde c_{i,t}$, then $c_{i,t}>\max\{g_{i,t},m_{i,t}\}$, so $z_{i,t}>0$ and the local slope is $\pi^{+}$: lower $c_{i,t}$ as in Case 1. If $c_{i,t}<\tilde c_{i,t}$, then $c_{i,t}<\min\{g_{i,t},M_{i,t}\}$, so $z_{i,t}<0$ and the local slope is $\pi^{-}$: raise $c_{i,t}$ as in Case 2. Either case preserves the sign of $z_{i,t}$ until the target is reached, so no step increases $J$.

Hence every completing sequence is dominated on every realization by the asserted action, which depends only on $\bm\psi_t$ and the current state.

\end{document}